%% file: 0.tex
\newif\ifams
\amstrue
\ifams
\documentclass[runningheads,envcountsame,reqno]{amsart}
\usepackage{amsaddr}
\usepackage{amsthm}
\else
\documentclass[runningheads,envcountsame]{llncs}
\fi

\usepackage{amsmath}
\usepackage{amssymb}
\usepackage{hyperref}
\def\citecolor{blue}
\hypersetup{colorlinks=true, citecolor=\citecolor, urlcolor=\citecolor,
  linkcolor=\citecolor}
\usepackage{tikz}
\usepackage{bibnames}
\usetikzlibrary{calc,matrix,cd}

\input{macros}

\title{The \iq of \\\cdlatt{s}}

\begin{document}
\ifams
\input{amsFrontMatter}
\else
\input{llncsFrontMatter}
\fi
\maketitle

\begin{abstract}
  \input{abstract}

\end{abstract}
\input{intro}

\input{elements}

\input{invquantaloids}

\input{lookingForZero}

\input{center}

\input{raney}

\input{iqcdlatts}

\input{eqTheory}

\input{conclusions}

\bibliographystyle{splncs04}
\bibliography{biblio}

\end{document}

%% file: amsFrontMatter.tex
\title{The involutive quantaloid of completely distributive lattices}
\author{Luigi
  Santocanale}%
\address{%
  Laboratoire d'Informatique et des Syst\`emes, \\
  UMR 7020, Aix-Marseille
  Universit\'e, CNRS }
\email{luigi.santocanale@lis-lab.fr}
\thanks{Partially supported
  by the ``LIA LYSM AMU CNRS ECM INdAM'', by the ``LIA LIRCO'', and
  by the ANR project TICAMORE}


%% file: abstract.tex
Let $L$ be a complete lattice and let $\QL$ be the unital quantale of
\jc endo-functions of $L$. We prove that $\QL$ has at most two cyclic
elements, and that if it has a non-trivial cyclic element, then $L$ is
\cd and $\QL$ is involutive (that is, non-commutative cyclic
$\star$-autonomous). If this is the case, then the dual tensor
operation corresponds, via Raney's transforms, to composition in the
(dual) quantale of \mc endo-functions of $L$.

Let $\SLatt$ be the category of sup-lattices and \jc functions and let
$\CDLatt$ be the full subcategory of $\SLatt$ whose objects are the
\cdlatt{s}.  We argue that $\CDLatt$ is itself an involutive
quantaloid, thus it is the largest full-subcategory of \SLatt with
this property. Since $\CDLatt$ is closed under the monoidal operations
of $\SLatt$, 
we also argue that if $\QL$ is involutive, then $\QL$ is \cd as well;
consequently, any lattice embedding into an involutive quantale of the
form $\QL$ has, as its domain, a distributive lattice.

%% file: intro.tex
\section{Introduction}

Let $C$ be a finite chain or the unit interval of the reals. In a
series of recent works \cite{2018-RAMICS,CW,2019-WORDS} we argued
that the unital quantale structure of $\QC$, the set of \jc functions
from $C$ to itself, plays a fundamental role to solve more complex
combinatorial and geometrical problems arising in Computer Science.
In \cite{2018-RAMICS,CW} we formulated an order theoretic approach to
the problem of constructing discrete approximations of curves in
higher dimensional unit cubes.  On the side of combinatorics, the
results in \cite{2019-WORDS} yield bijective proofs of counting
results (that is, bijections, through which these results can easily
be established) for idempotent monotone endo-functions of a finite
chain \cite{Howie71,LaradjiUmar2006} and a new algebraic
interpretation of well-known combinatorial identities
\cite{Worpit2011}.

The quantales $\QC$, $C$ a finite chain or $[0,1] \subseteq \R$, are
\emph{involutive}---or, using another possible naming,
\emph{non-commutative cyclic $\star$-autonomous}. The involution is
used in the mentioned works, even if is not clear to what extent it is
necessary. It was left open in these works whether there are other
complete chains $C$ such that $\QC$ is involutive and, at its
inception, the aim of this research was to answer this question.
Let us use $\QL$ for the unital quantale of \jc endo-functions of a
complete lattice $L$.
Recalling that \iqe structures on a given quantale are determined by
the cyclic dualizing elements and that complete chains are \cd, the
following statement from the monograph \cite{EGHK2018} shows that
$\QC$ is involutive for each complete chain $C$.
\begin{prop}{\propno in \cite{EGHK2018}}
  Let $L$ be a complete lattice and let $o_{L}$ be the \jc self-mapping
  on $L$ defined by $o_{L}(x) \eqdef \bigvee_{\Cond{z}{x}} z$, for
  $x \in L$. Then the following assertions are equivalent: (i) $o_{L}$
  is a dualizing element of the quantale $\Q(L)$, (ii) $L$ is \cd.
\end{prop}
This proposition also covers another important example studied in the
literature. Let $\D(P)$ be the perfect \cdlatt of downsets of a poset
$P$. According to the proposition, $\QDP$ is involutive, a fact that
can also be inferred via the isomorphism with the residuated lattice
of weakening relations on $P$, known to be involutive, see
\cite{Rosenthal1992,RSDLL,Jipsen2017}.

We strengthen here the above statement in many ways.  Firstly, we
observe that the quantale $\QL$ has at most two cyclic elements and
that cyclicity of $o_{L}$ is almost sufficient for $\QL$ to be
involutive:
\begin{thm}
  If $c \in \QL$ is cyclic, then either $c$ is the top element of
  $\QL$ or $c = o_{L}$.  Moreover, if $o_{L}$ is cyclic and not equal
  to the top element of $\QL$, then $L$ is \cd (and therefore $\QL$ is
  involutive, as from Propoosition~\propno of \cite{EGHK2018}).
\end{thm}
An important consequence of the previous statement is that the
quantale $\QL$ can be made into an \iqe in a unique way:
\begin{thm}
  If the quantale $\QL$ 
  is involutive,
  then its dualizing cyclic element is the \jc function $o_{L}$.
\end{thm}
In the direction from $L$ to $\QL$, we observe that the local \iqe
structures on each \cdlatt fit together in a uniform way.  A
quantaloid is a category whose homsets are \clatt{s} and for which
composition distributes on both sides with suprema. As a quantale can
be considered as a one-object quantaloid, the notion of \iqe naturally
lifts to the multi-object context---so an \iqe is a one-object \iq.
\emph{\Iq{s}} are indeed the Girard quantaloids introduced in
\cite{Rosenthal1992}. The following statement, proved in this paper,
makes precise the intuition that the local \iqe structures are
uniform:
\begin{thm}
  The full subcategory of 
  the category of complete lattices
  and \jc functions
  whose objects are the \cdlatt{s} is an \iq.
\end{thm}
The tools used in this paper rely on and emphasize Raney's
characterization of \cdlatt{s} \cite{Raney53,Raney60}.  A main remark
that we develop is that if $\QL$ is involutive, then the dual quantale
structure of $\QL$ 
arises from $\Q(L^{\partial})$, the quantale of \mc endo-functions of
$L$, via Raney's transforms (to be studied in
Section~\ref{sec:raney}).

Overall, this set of results yields an important clarification of the
algebra used in our previous works
\cite{2018-RAMICS,CW,2019-WORDS} 
and, more importantly, new characterizations of \cdlatt{s} adding up
to the existing ones, see e.g.
\cite{Raney53,Raney60,Lambrou83,HiggsRowe1989,Stubbe07}.  These
characterizations strongly rely on the algebra of quantales and \rl{s}
thus on relation algebra, in a wider sense.

An ideal goal of future research is to characterize the equational
theory of the \irl{s} of the form $\QL$.  For the moment being, we
observe that the units of the \iqe $\QL$, $L$ a \cdlatt, may be used
to characterize properties of $L$:
\begin{thm} 
  A \clatt is a chain if and only if the inclusion $0 \leq 1$ (in the
  language of \irl{s}) holds in $\QL$, i.e. if and only if $\QL$
  satisfies the mix law.  A \cdlatt has no \cjp elements if and only
  if the inclusion $1 \leq 0$ holds in $\QL$.
\end{thm}
It is known that the full subcategory of the category of complete
lattices and \jc functions whose objects are the \cdlatt{s}, the \iq
of \cdlatt, is closed under the monoidal operations inherited from the
super category, see e.g. \cite{Rowe1988,HiggsRowe1989,EGHK2018}.  In
particular, this quantaloid is itself a $\star$-autonomous category.
For the sake of studying the equational theory of the $\QL$, this fact
and the previous results jointly yield the following remarkable
consequence:
\begin{cor}
  If $\QL$ is an \iqe, then it is \cd.
\end{cor}

On the side of logic, it is worth observing that enforcing a linear
negation (the involution, the star) on the most typical models of
intuitionistic non-commutative linear logic also enforces a classical
behaviour, distributivity, of the additive logical connectors.  Apart
from the philosophical questions about logic, the above corollary
pinpoints an important obstacle in finding Cayley style representation
theorems for \irl{s} or a generalization of Holland's theorem
\cite{holland1963} from lattice-ordered groups to \irl{s}:
\begin{cor}
  If a \rl embedding of $\Q$ into some \irl of the form $\QL$ exists,
  then $\Q$ is distributive.
\end{cor}
Finally, we observe that these mathematical results pinpoint the
importance and the naturalness of considering a linear logic based on
a distributive setting. This algebraic setting has already many
established facets and applications. Among them, let us mention
bunched implication logic, for which our last theorem provides
non-standard pointless models. Let us also mention the usage of this
algebra in pointfree topology: here embeddability problems for
quantales dual to topological groupoids, problems analogous to the
ones we are raising, have already been investigated in depth, see
e.g. \cite{ProtinResende2012}.

\medskip

The paper is organised as follows.  In Section~\ref{sec:elements} we
provide definitions and elementary results.
In Section~\ref{sec:invquantaloids} we introduce the notion of an
involutive quantaloid (we shall identify an involutive quantale with
a one-object \iq).
We prove in Section~\ref{sec:lookingForZero} that if a quantale of the
form $\QL$ is involutive, then it has just one cyclic dualizing
element.  That is, there can be at most one \iqe structure extending
the structure of $\QL$. Moreover, we prove in this section that if
$\QL$ has a non-trivial cyclic element, then $L$ is a \cdlatt.
The uniqueness of the involutive structure is intimately related to
the fact---analyzed at the end of
Section~\ref{sec:lookingForZero}---that the only central elements of
$\QL$ are the identity and the constant function mapping to the bottom
of $L$.
In Section~\ref{sec:raney} we introduce Raney's transforms and their
elementary properties. Raney's transforms are the main tool used to
prove, in Section~\ref{sec:iqantaloid}, that \cdlatt{s} form an \iq.
In Section~\ref{sec:eqTheory} we develop some considerations on the
equational theories of the lattices $\QL$ among which, the use of the
multiplicative units of $\QL$ to characterize properties of $L$ and
the fact that $\QL$ is \cd whenever it is $\QL$ involutive.

%% file: elements.tex
\section{Definitions and elementary results}
\label{sec:elements}

\paragraph{Complete lattices and the category \SLatt.}

A \emph{complete lattice} is a poset $L$ such that each $X\subseteq L$ has
a supremum $\bigvee X$.   A map $f : L \rto M$ is \emph{\jc} if
$f(\bigvee X) = \bigvee f(X)$, for each subset $X \subseteq L$.  We
shall denote by $\SLatt$ the category whose objects are the \clatt{s}
and whose morphisms are the \jc maps.

For a poset $P$, $\dual{P}$ denotes the poset with the same elements
of $P$ but with the reverse ordering: $x \leq_{\dual{P}}y$ iff
$y \leq_{P} x$.  In a complete lattice, the set
$\bigvee \set{y \mid y \leq x, \text{ for each $x \in X$}}$ is the
infimum of $X$.  Therefore, if $L$ is complete, then $\dual{L}$ is
also a complete lattice. Moreover, if $L,M$ are \clatt{s} and
$f : L \rto M$ is \jc, then the map $\ra{f} : M \rto L$, defined by
$\ra{f}(y) \eqdef \bigvee \set{x \in L \mid f(x) \leq y}$, preserves
infima and therefore it belongs to the homset
$\SLatt(\dual{M},\dual{L})$.  The map $\ra{f}$ is the \emph{\radj} of
$f$, meaning that, for each $x\in L$ and $y \in M$, $f(x) \leq y$ if
and only if $x \leq \ra{f}(y)$. For $g : M \rto L$ \mc, its
\emph{\ladj} $\la{g} : L \rto M$ is defined similarly, and satisfies
$\la{g}(x) \leq y$ if and only if $x \leq g(y)$, for each $x \in L$
and $y \in M$. Consequenly, $\la{\ra{f}} = f$ and $\ra{\la{g}} = g$.
Indeed, by defining with $\dual{f} \eqdef \ra{f}$,
$\DUAL : \SLatt \rto \SLatt^{op}$ is a (contravariant) functor and a
category isomorphism.

Let $\set{f_{i} \mid i \in I }$ be a family of \jc functions from $L$
to $M$. The function $\bigvee_{i \in I} f_{i}$, defined by
$ (\bigvee_{i \in I} f_{i})(x) \eqdef \bigvee_{i \in I} f_{i}(x)$,
is a \jc map from $L$ to $M$. Therefore the homset $\SLatt(L,M)$, with
the pointwise ordering, is a \clatt, where suprema are computed by the
above formula. The same formula shows that the inclusion of
$\SLatt(L,M)$ into $M^{L}$, the set of all functions form $L$ to $M$,
is \jc. It follows that, for every $f : L \rto M$, there is a
(uniquely determined) greatest \jc function $h \in \SLatt(L,M)$ such
that $h \leq f$; in the following we shall use $\int(f)$ to denote
such $h$. Observe also that, by monotonicity of composition,
$\int(g) \circ \int(f) \leq g \circ f$ and therefore
$\int(g) \circ \int(f) \leq \int(g \circ f)$.

\paragraph{Quantales and involutive quantales.} A \emph{quantale} is a
complete lattice $\Q$ coming with a semigroup operation $\circ$ that
distributes with arbitrary sups. That is, we have
$(\bigvee X) \circ (\bigvee Y) = \bigvee_{x \in X, y \in Y} x \circ
y$, for each $X, Y \subseteq \Q$. A quantale is unital if the
semigroup operation has a unit. As we shall always consider unital
quantales, we shall use the wording quantale as a synonym of unital
quantale. In a quantale $\Q$, left and right residuals are defined as
follows:
$x \below y \eqdef \bigvee \set{z \in \Q \mid x \circ z \leq y}$ and
$y \upon x \eqdef \bigvee \set{z \in \Q \mid z \circ x \leq y}$.
Clearly, we have the following adjointness relations:
$x \circ y \leq z$ iff $y \leq x \below z$ iff $x \leq z \upon y$.
Let us recall that a quantale $\Q$ is a \emph{\rl}, that is an algebra
on the signature $\land,\vee,1,\circ,\below,\upon$, satisfying a
finite identities, see e.g. \cite[\S 2.2]{GJKO}.

A standard example of quantale is $\QL$, the set of \jc endo-functions
of a complete lattice $L$. In this case, the \sg operation is function
composition; otherwise said, $\QL$ is the homset $\SLatt(L,L)$.
We shall consider special elements of $\QL$ and of $\Q(\dual{L})$.
For $x\in L$, let $c_{x},a_{x},\alpha_{x} : L \rto L$ be defined as
follows:
\begin{align}
  \label{eq:defCA}
  c_{x}(t) & \eqdef
  \begin{cases}
    x\,, & t \neq \bot \,,\\
    \bot \,, & t = \bot\,,
  \end{cases}
  &
  a_{x}(t) & \eqdef
  \begin{cases}
    \top\,, & t \not\leq x\,, \\
    \bot \,, & t \leq x\,,
  \end{cases}
  &
  \alpha_{x}(t) & \eqdef
  \begin{cases}
    \top\,, & x \leq t\,, \\
    \bot \,, & x \not \leq t\,.
  \end{cases}
\end{align}
Clearly, $c_{x}, a_{x} \in \QL$ while $\alpha_{x} \in
\Q(\dual{L})$. Moreover, we have $\ra{c_{x}} = \alpha_{x}$.

\paragraph{Completely distributive lattices.}
A \clatt $L$ is said to be \emph{\cd} if, for each pair of families
$\pi : J \rto I$ and $x : J \rto L$, the following equality holds
\begin{align*}
  \bigwedge_{i \in I} \bigvee_{j \in J_{i}} x_{j} & = \bigvee_{\psi} \bigwedge_{i \in I} x_{\psi(i)}\,,
\end{align*}
where $J_{i} = \pi^{-1}(i)$, for each $i \in I$, and the meet on the
right is over all sections $\psi$ of $\pi$, that is, those functions
such that $\pi \circ \psi = id_{I}$. Let us recall that the notion of
a \cdlatt is auto-dual, meaning that a \clatt $L$ is \cd iff
$\dual{L}$ is such. For each \clatt $L$, define
\begin{align}
  \label{eq:defOOmega}
  o_{L}(x) & \eqdef \bigvee \set{t \mid \COND{x}{t}}\,,
  & \omega_{L}(y) & \eqdef \bigwedge \set{t \mid \COND{t}{y}}\,.
\end{align}
It is easy to see that $o_{L} \in \QL$ and that
$\ra{o_{L}} = \omega_{L}$. The following statement appears in
\cite[Theorem~4]{Raney60}:
\begin{theorem}[Raney]
  \label{thm:raney}
  A lattice is \cd if and only if any of the following equivalent
  conditions hold:
  \begin{align}
    \label{eq:omegao}
    \bigvee_{\COND{x}{t}} \omega_{L}(t)
    &  = x\,,
    & \bigwedge_{\COND{t}{y}} o_{L}(t)
    &  = y\,.
  \end{align}
\end{theorem}

%% file: invquantaloids.tex
\section{Involutive quantaloids}
\label{sec:invquantaloids}

The purpose of this section is to define \emph{\iq{s}} which, not
surprisingly, turn out to be the \emph{Girard quantaloids} of
\cite{Rosenthal1992}.  Let us mention that, following
\cite{Barr1979,Barr1995} and \cite{2018-RAMICS,CW}, another possible
naming for the same concept is \emph{non-commutative, cyclic,
  star-autonomous quantaloid}.  For the sake of conciseness, we prefer
the wording \iq.

We recall that a \emph{quantaloid}, see e.g. \cite{Stubbe07}, is a
category $\Q$ enriched over the category of sup-lattices. This means
that, for each pair of objects $L,M$ of $\Q$, the homset $\Q(L,M)$ is
a \clatt and that composition distributes over suprema in both
variables,
$(\bigvee_{i \in I} g_{i}) \circ (\bigvee_{j \in J} f_{j}) =
\bigvee_{i \in I, j \in J} f_{i} \circ g_{j}$. A quantale, see
e.g. \cite{Rosenthal1990}, might be seen as a one-object
quantaloid. The category $\SLatt$ is itself a quantaloid.
The definition below mimics, in a multisorted setting, a possible
definition of \iqe or of \irl.  For the possible equivalent
definitions of these notions, see e.g. \cite{Barr1995} or \cite[\S
3.3]{GJKO}.
\begin{definition}
  An \emph{\iq} is a quantaloid $\Q$ coming with operations
  \begin{align*}
    \fun{\star_{L,M}} : \Q(L,M) \rto \Q(M,L)\,,
    \text{\qquad $L,M$ objects of $\Q$,}
  \end{align*}
  satisfying 
  the following conditions:
  \begin{enumerate}
  \item $(f^{\star_{L,M}}){}^{\star_{M,L}} = f$, for each
    $f \in \Q(L,M)$,
  \item for each $f,g \in \Q(L,M)$,
    \begin{align*}
      f & \leq g \quad\tiff\quad f \circ g^{\star_{L,M}} \leq 0_{M}
      \quad\tiff\quad g^{\star_{L,M}} \circ f \leq 0_{L}\,,
    \end{align*}
    where $0_{M} \eqdef (id_{M})^{\star_{M,M}}$ and
    $0_{L} \eqdef (id_{L})^{\star_{L,L}}$.
  \end{enumerate}
  An \emph{\iqe} is a one-object \iq. 
\end{definition}
The superscripts $L$ and $M$ in $\fun{\star_{L,M}}$ shall be omitted
if they are clear from the context.
We state next elementary facts without proofs, 
the reader shall have no difficulty providing them.
For a category $\C$ enriched over posets, we use $\C^{co}$ for the
category with same objects and homsets, but for which the order is
reversed.
\begin{lemma}
  In an \iq $\Q$,
  if any of the inequalities below holds,
  then so do the other two:
  \begin{center}
    \commTriangle{L}{M}{N}{f}{g}{h} 
    \commTriangle{N}{L}{M}{h^{\star}}{f}{g^{\star}} 
    \commTriangle{M}{N}{L}{g}{h^{\star}}{f^{\star}} 
  \end{center}
  In particular, the operations $\star$ are order reversing, so $\star$
  is the arrow part of a functor $\Q \rto (\Q^{op})^{co}$ which is the
  identity on objects.
\end{lemma}

Let us recall that in any quantaloid residuals 
exist being defined as follows: for $f : L \rto M$, $g : M \rto N$,
and $h : L \rto N$, 
\begin{align*}
  g \below h : L \rto[] M & \eqdef \bigvee \set{ k 
    \mid g \circ k \leq h } \,,
  \;\;\;
  h \upon f : M \rto[] N 
  \eqdef \bigvee \set{ k 
    \mid k \circ f \leq h }\,,
\end{align*}
so, the usual adjointness relations hold:
$g \circ f \leq h$ iff
$f \leq g \below h$ iff
$g \leq h \upon f$.
\begin{lemma}
  In an \iq, for $f : L \rto M$, $g : M \rto N$, and $h : L \rto N$,
  we have the following equalities:
  \begin{align*}
    g \below h & = (h^{\star_{L,N}} \circ g)^{\star_{M,L}}\,, 
    &  
    h \upon f & = (f \circ h^{\star_{L,N}})^{\star_{N,M}}\,.  
  \end{align*}
  In particular (for $L = N$ and $h = 0_{L}$) we have
  $g \below 0_{L} = g^{\star_{M,L}}$ and
  $0_{L} \upon f = f^{\star_{L,M}}$.
\end{lemma}

Let us argue that our definition coincides with the definition of a
Girard quantaloid given in \cite{Rosenthal1992}.  It is readily seen
that, given an \iq $\Q$, the collection
$\set{0_{L} = id_{L}^{\star} \mid \text{ $L$ an object of $\Q$ } }$ is
a cyclic dualizing family in the sense of \cite{Rosenthal1992}.
Conversely, given such a family and $f : L \rto M$, we can define
$f^{\star} \eqdef f \below 0_{M}$ and this definition yields an \iq
structure as defined here. This definition also sets a bijective
correspondence between the two kind of structures.

%% file: lookingForZero.tex
\section{Cyclic elements of $\QL$
}
\label{sec:lookingForZero}

We prove in this section that if a \qe of the form $\QL$ is
involutive, then $id_{L}^{\star}$ equals $o_{L}$ defined in
equation~\eqref{eq:defOOmega}. From this it follows that there is at
most one \iqe structure on $\QL$ extending the \qe structure.
Moreover, we also prove that if $o_{L}$ is cyclic and distinct from
$c_{\top}$, then $L$ is \cd.
To this end, let us firstly recall the following standard definitions:
\begin{definition}
  Let $\Q$ be a quantale.  An element $\alpha \in \Q$ is said to be
  \begin{itemize}
  \item \emph{cyclic} if $f \below \alpha = \alpha \upon f$, for each
    $f \in \Q$,
  \item \emph{dualizing} if
    $(\alpha \upon f) \below \alpha = \alpha \upon (f \below \alpha) =
    f$, for each $f \in \Q$.
  \end{itemize}
\end{definition}
We already mentioned that \iqe structures on a \qe $\Q$ are in
bijection with cyclic dualizing elements of $\Q$.
Let us also recall that, 
for an \iq $\Q$ and an object $L$ of $\Q$,
$0_{L} := (id_{L})^{\star_{L,L}}$ is both a cyclic and a dualizing
element of the quantale $\Q(L,L)$.

\medskip

An important first observation, stated in the next lemma, is that
residuals of the form $g\below h$ in $\SLatt$ can be constructed by
means of the operations 
$\int(\intfun)$ (greatest \jc map below a given one) and
$\ra{\intfun}$ (taking the \radj of a \jc map).
\begin{lemma}
  For each 
  $g \in \SLatt(M,N)$, $h \in \SLatt(L,N)$, we have
  \begin{align*}
    g \below h & = \int(\rho(g) \circ h).
  \end{align*}
\end{lemma}
\begin{Proof}
  Indeed, for each $f \in \SLatt(L,M)$, we have $f \leq g \below h$
  \tiff $g \circ f \leq h$, \tiff $g(f(x)) \leq h(x)$, for each
  $x \in L$, 
  \tiff $f(x) \leq \ra{g}(h(x))$, for each $x \in L$, \tiff
  $f \leq \rho(g) \circ h$, \tiff $f \leq \int(\rho(g) \circ h)$.
\end{Proof}

For the next lemma, recall that the \jc map $o_{L}$ has been defined
in~\eqref{eq:defOOmega} and that the maps $c_{t}$ and $a_{t}$ have
been defined in~\eqref{eq:defCA}.
\begin{lemma}
  \label{lemma:oLAsAJoin}
  We have $o_{L} = \bigvee_{t \in L} c_{t} \circ a_{t}$. 
\end{lemma}
\begin{proof}
  Observe that $c_{t}(a_{t}(x)) = \bot$, if $x \leq t$, and
  $c_{t}(a_{t}(x)) = t$, if $x \not\leq t$. Therefore
  \begin{align*}
    (\bigvee_{t \in L} c_{t} \circ a_{t})(x) & = \bigvee_{t \in L}
    (c_{t} (a_{t})(x)) = \bigvee_{\overset{t \in L,}{c_{t}(a_{t}(x))
        \neq \bot}} c_{t} (a_{t}(x)) = \bigvee_{\overset{t \in L,}{x
        \not\leq t}} t = o_{L}(x)\,.  \tag*{\qedhere}
  \end{align*}
\end{proof}

\begin{lemma}
  \label{lemma:intAlpha}
  For each $x \in L$, 
  $\int(\alpha_{x}) = a_{o_{L}(x)}$.
\end{lemma}
\begin{Proof}
  Let us observe that $a_{o(x)} \leq \alpha_{x}$. This amounts to
  veryfing that if $\alpha_{x}(t) = \bot$, then $a_{o(x)}(t) =
  \bot$. Now, $\alpha_{x}(t) = \bot$ \tiff $x \not\leq t$, and so
  $t \leq o(x)$, thus $a_{o(x)}(t) = \bot$.
  Next, 
  let us suppose that $f : L \rto L$ is \jc and below
  $\alpha_{x}$. Thus, if $\alpha_{x}(t) = \bot$, that is, if
  $x \not\leq t$, then $f(t) = \bot$.  Then
  $f(o(x))=f(\bigvee_{x \not \leq t} t) = \bigvee_{x \not \leq t} f(t)
  = \bot$. By monotonicity of $f$, if $t \leq o(x)$, then
  $f(t) = \bot$, showing that $f \leq a_{o(x)}$.
\end{Proof}

\begin{theorem}
  \label{thm:twocyclic}
  For each \clatt $L$, the \qe $\QL$ has at most two cyclic elements,
  among $c_{\top}$ and $o_{L}$.
\end{theorem}
\begin{Proof}
  Now, let $h \in \QL$ be cyclic.
  First we  prove that $o_{L} \leq h$.  Consider that, for each
  $x \in L$, $a_{x} \circ c_{x} = c_{\bot} \leq h $. Thus, since
  $g \circ f \leq h$ if and only if $f \circ g \leq h$, we also have
  $c_{x} \circ a_{x} \leq h$. Since this relation holds for each
  $x \in L$, then, using Lemma~\ref{lemma:oLAsAJoin}, the relation
  $o_{L} = \bigvee_{x \in L} c_{x} \circ a_{x} \leq h$ holds.

  We argue now that if $h \neq c_{\top}$, then $h \leq o_{L}$ and
  therefore $h = o_{L}$.
  Let $x \in L$ and consider that $c_{x} \circ c_{x} \below h \leq
  h$. By cyclicity, we also have $c_{x} \below h \circ c_{x}\leq h$.

  Now,
  $c_{x} \below h = \int(\rho(c_{x}) \circ h) = \int(\alpha_{x} \circ
  h)$ and therefore, using Lemma~\ref{lemma:intAlpha},
  \begin{align*}
    a_{o_{L}(x)} \circ h \circ c_{x} & = \int(\alpha_{x}) \circ \int(h)
    \circ c_{x} \leq \int(\alpha_{x} \circ h) \circ c_{x} =
    c_{x} \below h \circ c_{x} \leq
    h\,.
  \end{align*}
  If $t \neq \bot$, then, by evaluating the above inequality at $t$,
  we get $a_{o_{L}(x)}(h(x)) \leq h(t)$.  Since $a_{o_{L}(x)}(h(x))$
  takes values $\bot$ and $\top$, this means that
  $a_{o_{L}(x)}(h(x)) = \top$ implies $\top \leq h(t)$, for all
  $t \neq \bot$. That is, if $h(x) \not\leq o_{L}(x)$, then
  $h(t) = \top$, for all $t \neq \bot$ and $x \in L$. Otherwise
  stated, if $h \not\leq o_{L}$, 
  then
  $h = c_{\top}$.
\end{Proof}

Let us recall that a nucleus on a quantale $\Q$ is a closure operator
$j$ such that $j(g) \circ j(f) \leq f(g \circ f)$. Nuclei are sort of
congruences in the category of quantales while quotients into some
\iqe bijectively correspond to nuclei $j$ of the form
$j(f) = (f \below 0) \below 0$ where $0$ is cyclic \cite[Theorem
1]{Rosenthal1990b}. Thus, the above theorem exhibits the quantales
$\Q(L)$ as sort of simple w.r.t. \iqe{s}.

\begin{lemma}
  \label{lemma:cTopNotDualizing}
  If $L$ 
  is not trivial, then $c_{\top}$ is not a dualizing element of $\QL$.
\end{lemma}
\begin{Proof}
  Observe that $c_{\top}$ is the greatest element of $\QL$ and, for
  this reason, $f \below c_{\top} = c_{\top} \upon f = c_{\top}$, for
  each $f \in \QL$. If $c_{\top}$ is dualizing, then
  $c_{\bot} = (c_{\top}\upon c_{\bot})\below c_{\top} =
  c_{\top}$. Considering that the mapping from sending $x \in L$ to
  $c_{x} \in \QL$ is an embedding, this shows that $\bot =\top$ in
  $L$.
\end{Proof}

\begin{corollary}
  If $h\in \QL$ is a cyclic and dualizing element, then $h = o_{L}$.
  That is, if $\QL$ is an \iqe, then $id_{L}^{\star} = o_{L}$.
\end{corollary}
\begin{Proof}
  If $L$ is trivial, then so is $\QL$, and $h = c_{\bot} = o_{L}$.
  If $L$ is not trivial, then, by Theorem~\ref{thm:twocyclic},
  $h \in \set{o_{L}, c_{\top}}$ and, by
  Lemma~\ref{lemma:cTopNotDualizing}, $h \neq c_{\top}$.
\end{Proof}

With respect to Theorem~\ref{thm:twocyclic}, we notice that
$c_{\top}$, being the top element of $\QL$, is always cyclic. It is
therefore pertinent to ask when $o_{L}$ is cyclic. Of course, this
is the case if $o_{L} = c_{\top}$.

\begin{theorem}
  \label{thm:direct}
  If $o_{L}$ is a cyclic element of $\QL$ and  $o_{L} \neq c_{\top}$, then
  $x = \bigwedge_{t \not\leq x} o_{L}(t)$,
  for each $x \in L$. Consequently, $L$ is a \cdlatt.
\end{theorem}
\begin{Proof}
  Since $o_{L}$ is cyclic, then, for each $x,y \in L$, the two
  conditions (a) $c_{y} \circ a_{x} \leq o_{L}$ and
  (b) $a_{x} \circ c_{y} \leq o_{L}$ are equivalent.

  Condition (a) states that, for each $t \in L$, $\COND{t}{x}$ implies
  $y \leq o_{L}(t)$; that is
  $y \leq \bigwedge_{\COND{t}{x}} o_{L}(t)$.  Condition (b) states
  that, for each $t \neq \bot$, if $\COND{y}{x}$ then
  $o_{L}(t) = \top$. This condition is equivalent to $\COND{y}{x}$
  implies $o_{L} = c_{\top}$ or, equivalently, to
  $o_{L} \neq c_{\top}$ implies $y \leq x$.  Thus we have that, if
  $o_{L}$ is cyclic and $o_{L} \neq c_{\top}$, then (c) for each
  $x,y \in L$, $y \leq x$ iff
  $y \leq \bigwedge_{\COND{t}{x}} o_{L}(t)$.
  Now, condition (c) is easily recognized to be equivalent to
  the equality $x = \bigwedge_{\COND{t}{x}} o_{L}(t)$, holding for
  each $x \in L$.
  From the latter identity, complete distributivity of $L$ follows
  using Raney's characterization of complete distributivity,
  Theorem~\ref{thm:raney}.
\end{Proof}

In this way we also obtain a refinement of one side of the equivalence
stated in Proposition~\propno of \cite{EGHK2018}, where we do not need
to refer to the cyclic dualizing element.
\begin{corollary}
  If $\QL$ is an \iqe, then 
  $L$ is a \cdlatt.
\end{corollary}
\begin{Proof}
  If $\QL$ is an \iqe, then its dualizing cyclic element is,
  necessarily, $o_{L}$. In particular, $o_{L}$ is cyclic and distinct
  from $c_{\top}$. By Theorem~\ref{thm:direct}, $L$ is \cd.
\end{Proof}
We shall see that $o_{L}$ is also dualizing if $L$ is \cd.  A
remarkable fact arising from these considerations is that, on the
class of pointed residuated lattices $\langle \QL,p \rangle$ (where
$p \in \QL$ is the point), the universal sentence
$p \neq \top \;\&\; \forall x. x \below p = p \upon x$ implies distributivity
as well as the linear double negation principle,
$x = (x \below p) \below p$.

%% file: center.tex
\subsubsection{The center of $\QL$.}
\label{sec:center}

Uniqueness of an \iqe structure extending the \qe structure of $\QL$
can also be achieved through the observation that the unique central
elements of $\QL$ are $id_{L}$ and $c_{\bot}$. We are thankful to
Claudia Muresan for her help with investigating the center of $\QL$.
\begin{definition}
  We say that an element $\beta$ of a quantale $\Q$ is 
  \begin{itemize}
  \item \emph{central} if $\beta \circ x = x \circ \beta$, for each
    $x \in \Q$,
  \item \emph{codualizing} if $x = \beta \below (\beta \circ x)$, for
    each $x \in \Q$.
  \end{itemize}
\end{definition}

\begin{lemma}
  If $\Q$ is an \iqe, then $\alpha \in \Q$ is cyclic if and only if
  $\alpha^{\star}$ is central and it is dualizing if and only if
  $\alpha^{\star}$ is codualizing.
\end{lemma}
\begin{Proof}
  Since $x \below \alpha = (\alpha^{\star} \circ x)^{\star}$,
  $\alpha \upon x = (x \circ \alpha^{\star})^{\star}$, and
  $\fun{\star}$ is invertible, the equality
  $x \below \alpha = \alpha \upon x$ holds if and only if the equality
  $\alpha^{\star} \circ x = x \circ \alpha^{\star}$ holds.

  Now $\alpha$ is dualizing if and only if, for each $x \in \Q$,
  $x = \alpha \upon (x \below \alpha) = \alpha^{\star}\below (x \below
  \alpha)^{\star} = \alpha^{\star}\below (\alpha^{\star} \circ x)$.
\end{Proof}

\begin{proposition}
  \label{prop:center}
  The only central elements of $\QL$ are $id_{L}$ and $c_{\bot}$.
\end{proposition}
\begin{Proof}
Clearly, $id_{L}$ and $c_{\bot}$ are central, so we shall be concerned
to prove that they are the only ones with this property.  To this end,
for $x_{0} \in L$, define
\begin{align*}
  \nu_{x_{0}}(t) & \eqdef
  \begin{cases}
    \bot\,, & t \leq x_{0}\,\\
    t\,, & \text{otherwise}.
  \end{cases}
\end{align*}
Notice that if $x_{0} = \bot$, then $\nu_{x_{0}} = id_{L}$, while if
$x_{0} = \top$, then $\nu_{x_{0}} = c_{\bot}$. 
  We firstly claim that if $\beta$ is central in $\QL$, then
  $\beta = \nu_{x_{0}}$, for some $x_{0} \in L$.
  Suppose $\beta$ is central. For each $x \in L$, we have
  $c_{x}(x) = x$ and therefore
  \begin{align*}
    \beta(x) & = (\beta \circ c_{x})(x) = c_{x}(\beta(x))\,.
  \end{align*}
  If $\beta(x) \neq \bot$, then, evaluating the rightmost expression,
  we obtain $\beta(x) = x$.  Let
  $x_{0} \eqdef \bigvee \set{y \mid \beta(y) = \bot}$, so
  $\beta(x_{0}) = \bot$. If $t \leq x_{0}$, then
  $\beta(t) \leq \beta(x_{0}) = \bot$ and, otherwise,
  $\beta(t) \neq \bot$ and so $\beta(t) = c_{t}(\beta(t)) =
  t$. Therefore, $\beta = \nu_{x_{0}}$.

  Next, we claim that if $x_{0} \not\in \set{\bot,\top}$, then
  $\nu_{x_{0}}$ is not central.
  Observe that
  \begin{align*}
    \nu_{x_{0}}(f(x)) &=
    \begin{cases}
      \bot\,, & f(x) \leq x_{0}\,,\\
      f(x)\,, & \text{otherwise}\,,
    \end{cases}
    &
    f(\nu_{x_{0}}(x)) &=
    \begin{cases}
      \bot\,, & x \leq x_{0}\,,\\
      f(x)\,, & \text{otherwise}\,.
    \end{cases} 
  \end{align*}
  It follows that if $\nu_{x_{0}} \circ f = f \circ \nu_{x_{0}}$, then
  $f(x_{0}) \leq x_{0}$.  Indeed, if $f(x_{0}) \not\leq x_{0}$, then
  $f(x_{0}) \neq \bot$, $\nu_{x_{0}}(f(x_{0})) = f(x_{0}) \neq \bot$,
  and $f(\nu_{x_{0}}(x_{0})) = \bot$.  Now, if
  $x_{0} \not\in \set{\bot,\top}$, then $c_{\top}$ is such that
  $x_{0} < \top = c_{\top}(x_{0})$, and therefore
  $\nu_{x_{0}} \circ c_{\top} \neq c_{\top} \circ \nu_{x_{0}}$.
\end{Proof}

It is now possible to argue that, for a complete lattice $L$, there
exists at most one extension of $\QL$ to an \iqe as follows.
Suppose that $\QL$ is involutive, so let $(\intfun)^{\star}$ be a
fixed \iqe structure. We shall argue that $id_{L}^{\star}$ is the
unique cyclic and dualizing element of $\QL$.
If $\alpha$
is an arbitrary cyclic and dualizing element of $\QL$, then
$\beta \eqdef \alpha^{\star}$ is central and codualizing and
$\beta \in \set{c_{\bot},id_{L}}$ using Proposition~\ref{prop:center}.
Since $\beta$ is codualizing, then it is an injective function: if
$\beta(x) = \beta(y)$, then $\beta \circ c_{x} = \beta \circ c_{y}$
and
$c_{x} = \beta \below (\beta \circ c_{x}) = \beta \below (\beta \circ
c_{y}) = c_{y}$; since the mapping sending $t$ to $c_{t}$ is an
embedding, we obtain $x = y$. Thus, if $L$ is not trivial,
$\beta \neq c_{\bot}$ (since $c_{\bot}$ is constant). Whether or not
$L$ is trivial, we derive $\beta = id_{L}$. It follows that
$\alpha = \alpha^{\star}{}^{\star} = \beta^{\star} = id_{L}^{\star}$.

%% file: raney.tex
\section{Raney's transforms}
\label{sec:raney}

Let $L, M$ be two complete lattices. For $f : L \rto M$, define
\begin{align*}
  \joinof{f}(x) & \eqdef \bigvee_{\Cond{t}{x}} f(t)\,,
  &
  \meetof{f}(x) & \eqdef \bigwedge_{\Cond{x}{t}} f(t)\,, \quad
  \text{for each $x \in L$.}
\end{align*}
We call $\joinof{f}$ and $\meetof{f}$ the \emph{Raney's transforms} of
$f$.  Notice that $f$ is not required to be monotone in order to
define $\joinof{f}$ or $\meetof{f}$ which, on the other hand, are
easily seen to be monotone; these functions are even join and \mc,
respectively, as argued in the next lemma.
\begin{lemma}
  For any $f : L \rto M$, define
  \begin{align}
    \label{eq:defgf}
    g_{f}(y) & \eqdef \bigwedge \set{z \mid \Cond{y}{f(z)}}\,.
  \end{align}
  Then $g_{f}$ is \radj to $\joinof{f}$ and therefore $\joinof{f}$ is
  \jc. Dually, $\meetof{f}$ is \mc.
\end{lemma}

We call the operation $\JOINOF$ Raney's transform for the following
reason.
For $\theta \subseteq L \times M$ an arbitrary relation, Raney
\cite{Raney60} defined (up to some dualities) 
\begin{align}
  r_{\theta}(x) & \eqdef \bigwedge \set{y \in M \mid \forall
    (t,v). \,(t,v) \in \theta \,\timplies\, x \leq t \,\tor\, v \leq y
  }\,.  \label{eq:raneysDef}
\end{align}
Recall that a \ladj $\ell : L \rto M$ can be expressed from its \radj
$\rho : M \rto L$ by the formula
$\ell(x) = \bigwedge\set{y \mid x\leq \rho(y)}$. Using this expression
with $\ell = \joinof{f}$ and $\rho = g_{f}$ defined
in~\eqref{eq:defgf}, we obtain
\begin{align}
  \label{eq:joinofViaraneysDef}
  \joinof{f}(x) & = \bigwedge \set{ y \in M \mid \forall t. \,f(t)
    \not\leq y\, \timplies \,x \leq t}\,.
\end{align}
Clearly, if in \eqref{eq:raneysDef} we let $\theta$ be the graph of
$f$, defined by $(t,v) \in \theta$ if and only if $f(t) = v$, then we
obtain equality between the right-hand sides of \eqref{eq:raneysDef}
and \eqref{eq:joinofViaraneysDef}, and so $\joinof{f} = r_{\theta}$.

We list next the few properties we need to know about these
transforms. 
\begin{lemma}
  \label{lemma:lemmaMon}
  \label{lemma:naturalOne}
  The transform $\funJoinOf$ has the following properties:
  \begin{enumerate}
  \item if $f \leq g : L \rto M$, then $\joinof{f} \leq \joinof{g}$,
  \item if $g : L \rto M$ and $f : M \rto N$ is monotone, then
    $\joinof{(f \circ g)} \leq f \circ (\joinof{g})$,
  \item if $g : L \rto M$ and $f : M \rto N$ is \jc, then
    $\joinof{(f \circ g)} = f \circ (\joinof{g})$,
  \item if $f : L \rto M$ is \jc (with $L$ and $M$ complete), then
   \begin{align}
    \label{eq:commUno}
    \la{\meetof{f}} & = \joinof{\ra{f}} : M \rto L\,.
  \end{align}
  \end{enumerate}
\end{lemma}
The proof of these properties does  not present difficulties, possibly
apart for the last item, for which we refer the reader to
\cite[Proposition 4.6 (b.iii)]{HiggsRowe1989}.

%% file: iqcdlatts.tex
\section{\CDLatt is an \iq}
\label{sec:iqantaloid}
We prove now that $\CDLatt$, the full subcategory of \SLatt whose
objects are the \cd lattices, is an \iq. By the results of
Section~\ref{sec:lookingForZero}, this is also the largest full
subcategory of \SLatt with this property.

Recall from Theorem~\ref{thm:raney} that a \clatt is \cd if and only
if $\joinof{\omega_{L}} = id_{L}$ (or, equivalently,
$\meetof{o_{L}} = id_{L}$).  
\begin{lemma}
  \label{lemma:interior}
  If $L$ is a \cdlatt and $f : L \rto M$ is monotone, then
  $\int(f) = \joinof{(f \circ \omega_{L})}$ and
  $\joinof{f} = \int(f \circ o_{L})$.
\end{lemma}
\begin{Proof}
  By monotonicity of $f$, we have
  $\joinof{(f \circ \omega_{L})} \leq f \circ (\joinof{\omega_{L}}) =
  f$.  Suppose that $g$ is \jc and $g \leq f$. Then
  $g = g \circ (\joinof{\omega_{L}}) = \joinof{(g \circ \omega_{L})}
  \leq \joinof{(f \circ \omega_{L})}$.
  To see that $\joinof{f} = \int(f\circ o_{L})$, observe that
  $\joinof{f}  = \joinof{(f \circ id)} \leq f \circ (\joinof{id_{L}})
  = f \circ o_{L}$,
  and therefore $\joinof{f} \leq \int(f \circ o_{L})$.  On the other
  hand,
  $\int(f \circ o_{L}) = \joinof{(f \circ o_{L} \circ \omega_{L})}
  \leq \joinof{f}$, using the conunit of the adjunction,
  $o_{L} \circ \omega_{L} \leq id_{L}$.
\end{Proof}

The interior operator so defined is quite peculiar, since
for $g : L \rto M $ monotone and $f : M \rto N$ \jc,
we have
\begin{align*}
  \int(f \circ g) & = \joinof{(f \circ g \circ \omega_{L})}
  = f \circ \joinof{(g \circ \omega_{L})} = f \circ \int(g)\,.
\end{align*}
In general, if $L$ is not a \cdlatt, then we would have, above, only
an inequality, since
$\int(f \circ g) \geq \int(f) \circ \int(g) = f \circ \int(g)$.

\begin{lemma}
  \label{lem:inverse}
  If $L$ is a \cdlatt and $f : L \rto M$ is \jc, then
  $f = \joinof{\meetof{f}}$.
\end{lemma}
\begin{Proof}
  We firstly show that $\joinof{\meetof{f}} \leq f$.
  If $\COND{x}{t}$, then $\meetof{f}(t) = \bigwedge_{\COND{u}{t}} f(u)
  \leq f(x)$ and therefore
  $\joinof{\meetof{f}}(x) = \bigvee_{\COND{x}{t}}
  \meetof{f}(t) \leq f(x)$, for all $x \in L$.
  Let us argue that $f \leq \joinof{\meetof{f}}$:
  \begin{align*}
    f & = f \circ id_{L} = f \circ (\joinof{\omega_{L}}) =
    \joinof{(f
      \circ \omega_{L})}
    \leq \joinof{\meetof{f}}\,,
  \end{align*}
  where we have used the fact, dual to the relation
  $\joinof{f} = \int(f \circ o_{L})$ established in
  Lemma~\ref{lemma:interior}, that $\meetof{f}$ is the least \mc
  function above $f \circ \omega_{L}$, so in particular
  $f \circ \omega_{L} \leq \meetof{f}$.
\end{Proof}

For $ f : L \rto M$ \jc, define $f^{\star_{L,M}} : M \rto L$ as
follows:
\begin{align*}
  f^{\star_{L,M}} & \eqdef \joinof{\rho(f)} = \ell(\meetof{f})\,.
\end{align*}
Let us remark that the mappings $\STAR$ so defined are the maps
witnessing that \cdlatt{s} are nuclear, see \cite[Theorem
4.7]{HiggsRowe1989}. We leave for future research to establish an
exact connection between the notions of \iq and of nuclear object in
an autonomous category.
\begin{theorem}
  \label{thm:converse}
  The operations $(\intfun)^{\star_{L,M}}$ so defined yield an \iq
  structure on \CDLatt.
\end{theorem}
\begin{Proof}
  Firstly, we verify that $\Star{\Star{f}} = f$ using
  Lemmas~\ref{lemma:naturalOne} and~\ref{lem:inverse}, and the fact
  that the \jc functions are in bijection with \mc functions via
  taking adjoints:
  $\Star{\Star{f}} = \joinof{\rho(\joinof{\rho(f)})} =
  \joinof{\rho(\ell(\meetof{f}))} = \joinof{\meetof{f}} = f$.

  We now verify that $\STAR$ satisfies the constraints needed to have
  an \iq. Let us remark that
  $id_{L}^{\star} = \joinof{\rho(id_{L})} = \joinof{id_{L}} = o_{L}$.

  Observe that since $\STAR$ is defined by composing an order
  reversing and an order preserving function, it is order
  reversing. Since it is an involution, then $f \leq g$ if and only if
  $g^{\star} \leq f^{\star}$.

  Now we assume that $f : L \rto M$ and $h : M \rto L$ and recall (see
  Lemma~\ref{lemma:interior}) that
  $h^{\star_{M,L}} = \joinof{\rho(h)} = \int(\rho(h) \circ o_{L}) : L
  \rto M$.  Therefore, $h \circ f \leq o_{L}$ if and only if
  $f \leq \rho(h) \circ o_{L}$, if and only if
  $f \leq \int(\rho(h) \circ o_{L}) = h^{\star}$.  Therefore, if
  $g : L \rto M$, then (letting $h = g^{\star}$) $f \leq g$ if and
  only if $g^{\star} \circ f \leq o_{L}$. Then, also, $f \leq g$ if
  and only if $g^{\star} \leq f^{\star}$ if and only if
  $ f \circ g^{\star} \leq o_{M}$.
\end{Proof}

Putting together Theorems~\ref{thm:direct} and~\ref{thm:converse}, we
obtain the following generalization of Proposition~\propno in
\cite{EGHK2018}, where no mention of the choice of the cyclic
dualizing element is required.
\begin{corollary}
  \label{cor:iff}
  The quantale $\QL$ is involutive if and only if $L$ is a \cdlatt.
\end{corollary}
For $f : L \rto M$ and $g : M \rto N$ (with $L,M,N$ \cdlatt{s}), let
us define
\begin{align*}
  g \oplus f & \eqdef (f^{\star} \circ g^{\star})^{\star} : L \rto M\,,
\end{align*}
and observe that
\begin{align*}
  (g \oplus f ) & = (f^{\star} \circ g^{\star})^{\star} =
  \joinof{\ra{(\la{\meetof{f}} \circ \la{\meetof{g}})}} =
  \joinof{(\meetof{g} \circ \meetof{f})}\,.
\end{align*}
That is:
\begin{proposition}
  The dual quantaloid structure arises via Raney's transforms from the
  composition
  $\SLatt(L^{\partial},M^{\partial}) \times
  \SLatt(M^{\partial},N^{\partial}) \rto
  \SLatt(L^{\partial},N^{\partial})$.
\end{proposition}

%% file: eqTheory.tex
\section{Remarks on the equational theory of the $\QL$}
\label{sec:eqTheory}

We develop in this section few considerations concerning the
equational theory of the \irl{s} $\QL$.  
\begin{theorem}
  \label{thm:mix}
  A \clatt $L$ is a chain if and only if $\QL$ is an \iqe satisfying
  the mix rule, i.e. the inclusion $x \circ y \leq x \oplus y$.
\end{theorem}
\begin{Proof}
  It is well known that the mix rule is equivalent to the inclusion
  $0 \leq 1$---where $1$ is the unit for $\circ$ and $0$ is the unit
  for $\oplus$.  Therefore, an \iqe of the form $\QL$ satisfies the
  mix rule if and only if $o_{L} \leq id_{L}$. This relation is easily
  seen to be equivalent to the statement that if $x \not\leq t$, then
  $t \leq x$, so $L$ is a chain. For the converse, we just need to
  recall that every chain is a \cdlatt.
\end{Proof}

Let us recall that an element $x$ of a \clatt $L$ is 
\emph{\cjp}
if, for every $Y \subseteq L$, the relation $x \leq \bigvee Y$ implies
$x \leq y$ for some $y \in Y$.  It is not difficult to see that $x$ is
\cjp if and only if $x \not\leq o_{L}(x)$.
Thus we say that a \clatt is \emph{smooth} if it has no \cjp
element. For example, the interval $[0,1]$ of the reals is a smooth
\cdlatt. The following statement is an immediate consequence of these
considerations.
\begin{theorem}
  \label{thm:comix}
  A \clatt $L$ is smooth if and only if $id_{L} \leq o_{L}$. Thus, a
  \cdlatt $L$ is smooth if and only if $\QL$ satisfies the inclusion
  $1 \leq 0$ in the language of \irl{s}.
\end{theorem}
These statements
generalize the remarks by \blue{Galatos and Jipsen, this collection,}
on the \irl of weakening relations on $P$. Let us recall that this
\irl is isomorphic to $\QDP$ where $\DP$ is the collection of downsets
of $P$. Thus, they observe that $\QDP$ satisfies the mix rule if and
only if $P$ is a chain, and that there are no non-trivial posets $P$
such that $\QDP$ satisfies the inclusion $1 \leq 0$. These facts might
be seen as consequences Theorems~\ref{thm:mix} and~\ref{thm:comix},
considering that $\DP$ is a chain if and only if $P$ is a chain, and
that $\DP$ is spatial, meaning that every element of $\DP$ is the join
of the \cjp elements below it (so, $\DP$ has plenty of \cjp elements).

\bigskip

For a family $\set{f_{i} \in \SLatt(L,M)\mid i \in I} $, let us define
$\bigwedge_{i \in I} f_{i}$ and $\bigWedge_{i \in I} f_{i}$ by
\begin{align}
  (\bigwedge_{i \in I} f_{i})(x) & \eqdef \bigwedge_{i \in I}
  (f_{i}(x))\,,\;\;
  \quad
  (\bigWedge_{i \in I} f_{i})(x) \eqdef \bigvee_{\COND{x}{t}}
  \bigwedge_{i \in I} f_{i}(\omega_{L}(t))\,.\;\;
  \label{eq:pointwiseFormula}
\end{align}
Notice that $\bigwedge_{i \in i} f_{i}$ need not be \jc while
$\bigWedge_{i \in I} f_{i}$ is \jc and
$\bigWedge_{i \in I} f_{i} = \int(\bigwedge_{i \in i} f_{i})$ if $L$
is \cd, see Lemma~\ref{lemma:interior}.  Under the latter condition,
$\bigWedge_{i \in I} f_{i}$ is the infimum of 
$\set{f_{i} \mid i \in I} $ within the complete lattice $\SLatt(L,M)$.
The explicit description of the infimum given in
\eqref{eq:pointwiseFormula} can be exploited to prove
that 
$\CDLatt$ 
is closed under the monoidal operations inherited from $\SLatt$, see
e.g.  \cite{Rowe1988,HiggsRowe1989,EGHK2018}, thus it is
$\star$-autonomous \cite{Barr1979}. We expect the formula in
\eqref{eq:pointwiseFormula} also to be useful for computational
issues, \blue{see Ramirez et al., this collection}.

Coming back to the equational theory of the $\QL$, an important
consequence of $\CDLatt$ being $\star$-autonomous is that $\QL$ is \cd
if $L$ is \cd (the converse holds as well). Then, the following
obstacle arises towards finding representation theorems for \irl{s}
via the $\QL$:
\begin{corollary}
  If an \irl $\Q$ has an embedding into an \iqe of the form $\QL$, then
  $\Q$ is distributive.
\end{corollary}
Indeed, if $\QL$ is an \iqe, then $L$ is a \cdlatt and $\QL$ as
well. 
Thus, if $\Q$ has a lattice embedding into $\QL$, then $L$ is
distributive.

%% file: conclusions.tex
\section{Conclusions and future steps}

The research exposed in this paper tackles and solves a natural
problem encountered during our investigations of certain \qe{s} built
from \cc{s} \cite{2019-WORDS,CW,2018-RAMICS}. The problem asks to
characterize the \cc{s} whose quantale of \jc endomaps is
involutive. Every \cc is a \cdlatt and by now we know that every \cc
has this property; in particular, other properties of chains and
posets, such as self-duality, are not relevant.

The solution provided, building on
\cite[Proposition~\propno]{EGHK2018}, is as general as possible, in
two respects. On the one hand, an exact characterization of all the
\clatt{s}---not just the chains---$L$ for which $\QL$ is involutive
becomes available: these are the \cdlatt{s}; improving on
\cite[Proposition~\propno]{EGHK2018}, we argue that the choice of a
cyclic dualizing element does not matter.  In particular, the
characterization covers different kind of \iqe{s} known in the
literature, those discovered in our investigation of complete chains
and those known as the \rl{s} of weakening relations---arising from
the relational semantics of distributive linear logic. On the other
hand, we show that the \iqe structures on \cdlatt{s} are uniform,
yielding and \iq structure on the category of \cdlatt{s} and \jc
functions.

We have drawn several consequences from the observations developed,
among them, the fact that if an \iqe $\Q$ can be embedded into an
quantale of the form $\QL$, then it is distributive.  This fact calls
for a characterization of the \irl{s} embeddable into some $\QL$, a
research track that might require to or end up with determining the
variety of \irl{s} generated by the $\QL$.
A second research goal, that we might tackle in a close future,
demands to investigate the algebra developed in connection with the
continuous weak order \cite{CW} in the wider and abstract setting of
\cdlatt{s}. Let us recall that in \cite{CW} a surprising bijection was
established between two kind of objetcs, the maximal chains in the
cube lattice $[0,1]^{d}$ and the families
$\set{f_{i,j} \in \Q([0,1])\mid 1 \leq i < j \leq d}$ such that, for
$i < j < k$,
$f_{j,k} \circ f_{i,j}\leq f_{i,k} \leq f_{j,k} \oplus f_{i,j}$. So,
are there other surprising bijections if the interval $[0,1]$ is
replaced by an arbitrary \cdlatt, and if we move from the \iqe setting
to the multisorted setting of \iq{s}?

\input{acks}

%% file: acks.tex
\ifams
\bigskip
\fi

\paragraph{Acknowledgment.} 
The author is thankful to Srecko Brlek, Claudia Muresan, and Andr\'e
Joyal for the fruitful discussions these scientists shared with him on
this topic during winter 2018.